\documentclass[]{rsos}

\usepackage[utf8]{inputenc}
\usepackage[english]{babel}
\usepackage{amsthm}
\usepackage{amsfonts}
\usepackage{amsmath}
\usepackage{yfonts}
\usepackage{authblk}
\usepackage{blindtext}
\usepackage{amssymb}
\usepackage{mathtools}

\usepackage[ruled,vlined,noresetcount]{algorithm2e}
\SetKwInOut{Parameter}{Parameters}
 
\theoremstyle{definition}
\newtheorem{definition}{\bf Definition}[section]

\newtheorem{theorem}{\bf Theorem}[section]

\newtheorem{proposition}{\bf Proposition}[section]

\DeclareMathOperator*{\argmaxA}{arg\,max}
\DeclareMathOperator*{\argminA}{arg\,min}
\newcommand{\Cov}{\mathrm{Cov}}

\begin{document}

\title{Risk-Aware Multi-Armed Bandit Problem with Application to Portfolio Selection}

\author{
Xiaoguang Huo$^{1}$ and Feng Fu$^{2,3}$}

\address{$^{1}$Department of Mathematics, Cornell University, Ithaca, NY 14850, USA\\
$^{2}$Department of Mathematics, Dartmouth College, Hanover, NH 03755, USA\\
$^{3}$Department of Biomedical Data Science, Geisel School of Medicine at Dartmouth, Lebanon, NH 03756, USA}

\subject{mathematical modeling, applied mathematics}

\keywords{multi-armed bandit, online learning, portfolio selection, graph theory, risk-awareness, conditional value-at-risk}

\corres{Xiaoguang Huo\\
\email{xh84@cornell.edu}\\
Feng Fu \\
\email{fufeng@gmail.com}\\}

\begin{abstract}
Sequential portfolio selection has attracted increasing interests in the machine learning and quantitative finance communities in recent years. As a mathematical framework for reinforcement learning policies, the stochastic multi-armed bandit problem addresses the primary difficulty in sequential decision making under uncertainty, namely the \textit{exploration} versus \textit{exploitation} dilemma, and therefore provides a natural connection to portfolio selection. In this paper, we incorporate risk-awareness into the classic multi-armed bandit setting and introduce an algorithm to construct portfolio. Through filtering assets based on the topological structure of financial market and combining the optimal multi-armed bandit policy with the minimization of a coherent risk measure, we achieve a balance between risk and return.
\end{abstract}


\begin{fmtext}
\section{Introduction}
Portfolio selection is a popular area of study in the financial industry ranging from academic researchers to fund managers. The problem involves determining the best combination of assets to be held in the portfolio in order to achieve the investor's objectives, such as maximizing the cumulative return relative to some risk measure. In the finance community, the traditional approach to this problem can be traced back to 1952 with Markowitz's seminal paper \cite{markowitz}, which introduces mean-variance analysis, also known as the modern portfolio theory (MPT), and suggests choosing the allocation that
\end{fmtext}
\maketitle
\noindent
maximizes the expected return for a certain risk level quantified by variance. On the other hand, sequential portfolio selection models have been developed in the mathematics and computer science communities. For example, Cover's universal portfolio strategy \cite{cover}, Helmbold's multiplicative update portfolio strategy \cite{helmbold}, and see Li \& Hoi \cite{survey} for a comprehensive survey. In recent years, with the unprecedented success of AI and machine learning methods evidenced by AlphaGo defeating the world champion and OpenAI's bot beating professional Dota players, more creative machine learning based portfolio selection strategies also emerged \cite{deep,sentiment}.

\indent
Including portfolio selection, many practical problems such as clinical trials, online advertising and robotics can be modeled as sequential decision making under uncertainty \cite{bayesian}. In such a process, at each trial the learner faces the trade-off between acting ambitiously to acquire new knowledge and acting conservatively to take advantage of current knowledge, which is commonly known as the \textit{exploration} versus \textit{exploitation} dilemma. Often understood as a single-state Markov Decision Process (MDP), the stochastic multi-armed bandit problem provides an extremely intuitive mathematical framework to study sequential decision making. 

An abstraction of this setting involves a set of $K$ slot machines and a sequence of $N$ trials. At each trial $t = 1,\dots,N$, the learner chooses to play one of the machines $I_t \in \{1,\dots,K\}$ and receives a reward $R_{I_t,t}$ drawn randomly from the corresponding fixed but unknown probability distribution $\nu_{I_t}$, whose mean is $\mu_{I_t}$. In the classic setting, the random rewards of the same machine across time are assumed to be independent and identically distributed, and the rewards of different machines are also independent. The objective of the learner is to develop a \textit{policy}, an algorithm that specifies which machine to play at each trial, to maximize cumulative rewards. A popular measure for the performance of a policy is the \textit{regret} after some $n$ trials, which is defined to be 
\begin{equation}
\xi(n) \stackrel{\text{def}}{=} \max_{\forall i \in [1,K]}\sum_{t = 1}^{n} R_{i,t} - \sum_{t = 1}^{n} R_{I_t,t}.
\end{equation}
However, in a stochastic model it is more intuitive to compare rewards in expectation and use \textit{pseudo-regret} \cite{regret}. Let $\mathnormal{T}_i(n)$ be the number of times machine $i$ is played during the first $n$ trials and let $\mu^* = \max\{\mu_1,\dots,\mu_K\}$. Then,
\begin{equation}
\widehat{\xi}(n) \stackrel{\text{def}}{=} n\mu^* - \mathbb{E}\sum_{t = 1}^{n}R_{I_t,t}  = \sum_{1\le i \le K,\, \mu_i < \mu^*}(\mu^* - \mu_i) \mathbb{E}[\mathnormal{T}_i (n)]
\end{equation}

Thus, the learner's objective to maximize cumulative rewards is then equivalent to minimizing regret. The asymptotic lower bound on the best possible growth rate of total regret is proved by Lai \& Robbins \cite{lai}, which is $\mathcal{O}(\log{n})$ with a coefficient determined by the suboptimality of each machine and the Kullback-Leibler divergence. Since then, various online learning policies have been proposed \cite{algorithm}, among which the UCB1 policy developed in Auer \textit{et al.} \cite{auer} is considered the optimal and will be introduced in detail in Section Methods and Model.

Although the classic multi-armed bandit has been well studied in academia, a number of variants of this problem are proposed to model different real world scenarios. For example, Agrawal \& Goyal \cite{agrawal} considers contextual bandit with a linear reward function and analyzes the performance of Thompson Sampling algorithm. Koulouriotis \& Xanthopoulos \cite{nonstationary} studies the non-stationary setting where the reward distributions of machines change at a fixed time. A more important variant is the risk-aware setting, where the learner considers risk in the objective instead of simply maximizing the cumulative reward. This variant is closely related to the portfolio selection problem, where risk management is an indispensable concern, and has been discussed in several papers. For example, Sani \textit{et al.} \cite{sani} studies the problem where the learner's objective is to minimize the mean-variance defined as $\sigma^2 - \rho\mu$ and proposes two algorithms, MV-LCB and ExpExp. In a similar setting, Vakili \& Zhao \cite{vakili} provides a finer analysis of the performance of algorithms proposed in Sani \textit{et al.} \cite{sani}. In addition, Vakili \& Zhao \cite{vakili2} extends this setting by considering the mean-variance and value-at-risk of total rewards at the end of time horizon. In a more generalized case, Zimin \textit{et al.} \cite{zimin} sets the objective to be a function of the mean and the variance $f(\mu,\sigma^2)$ and defines the $\varphi$-LCB algorithm that achieves desirable performance under certain conditions. Moreover, Galichet \textit{et al.} \cite{galichet} chooses the conditional value-at-risk to be the objective and proposes the MARAB algorithm.

These works serve as the inspiration for us to consider risk in the model, but they are not directly applicable to the portfolio selection problem, owing to the primary obstacle that these methods only choose the best single machine to play at each trial. To address this issue, a basket of candidate portfolios need to be first selected in the preliminary stage in a strategic and logical way. For example, Shen \textit{et al.} \cite{orthogonal} uses principle component analysis (PCA) to select candidate portfolios, namely the normalized eigenvectors of the covariance matrix of asset returns.\\
\indent
In our model, we first take a graph theory approach to filter and select a basket of assets, which we use to construct the portfolio. Then, at each trial we combine the single-asset portfolio determined by the optimal multi-armed bandit algorithm with the portfolio that globally minimizes a \textit{coherent} risk measure, the conditional value-at-risk. The rest of this paper is organized as follows. In Section Methods and Model, we formulate the portfolio selection problem in the multi-armed bandit setting, and describe our methodology in detail. In Section Results, we present our simulation results using the proposed method. In Section Discussions \& Conclusion we discuss results and also provide directions for future research.


\section{Methods and Model}
\subsection{Problem Formulation}
In this section, we modify the classic multi-armed bandit setting to model portfolio selection. Consider a financial market with a large set of assets, from which the learner selects a basket of $K$ assets to invest in a sequence of $N$ trials. At each trial $t = 1,\dots,N$, the learner chooses a portfolio $\boldsymbol{\omega}_t = {\left({\omega}_{1,t} , \dots , {\omega}_{K,t}\right)}^\top $ where $\omega_{i,t}$ is the weight of asset $i$. Since we only consider long-only and self-financed trading, we must have $\boldsymbol{\omega}_t \in W$ where $W = \{\boldsymbol{u} \in \mathbb{R}_{+}^{K} : \boldsymbol{u}^\top \boldsymbol{1} = 1\}$ and $\boldsymbol{1}$ is a column vector of ones. The returns of assets are then revealed at trial $t+1$ and denoted by $\boldsymbol{R}_{t} = {\left(R_{1,t} , \dots , R_{K,t}\right)}^\top$. In particular, the return for each asset $R_{i,t}$ is viewed as a random draw from the corresponding probability distribution $\nu_i$ with mean $\mu_i$ and can be simply defined as the log price ratio $R_{i,t} = \log\left(P_{i,t+1}/P_{i,t}\right)$, where we use the natural log and $P_{i,t}$, $P_{i,t+1}$ are the prices at trial $t$ and $t+1$. For the trading period from $t$ to $t+1$, the learner receives ${\boldsymbol{\omega}_t}^\top \boldsymbol{R}_{t}$ as the reward for his portfolio. The investment strategy of the learner is thus a sequence of $N$ mappings from the accumulated knowledge to $W$.\\
\indent
We make the following assumptions. First, we assume we always have access to historical returns $H_{i,t}$ of every asset $i$ in the market for $t = 1,\dots,\delta$. The historical return is defined similarly to  $R_{i,t}$ as the log of price ratio but corresponds to the time horizon immediately before our investment period. They are only used to estimate the correlation structure and risk level. Second, we make no assumption on the dependency of returns either across time or across assets. We only assume that for each trial $t$ and for all $i \in \{1,\dots,K\}$, $R_{i,t} \sim \nu_i$ and $H_{i,t} \sim \nu_i$ with a relatively small $\delta$. Note that the UCB1 algorithm we use later is proved to be optimal under a weaker assumption, $\mathbb{E}[R_{i,t} | R_{i,1}, \dots,R_{i,t-1}] = \mu_i$, allowing us to waive the assumptions in the classic setting \cite{auer}. Third, transaction costs and market liquidity will not be considered. See Model 1 for a summary of the problem.
\begin{algorithm}[!htb]
\SetAlgoLined
\Parameter{$\delta$, $N$}
\SetAlgorithmName{Model}{}

 Receive historical returns $H_{i,t}$ of each asset $i$ for $t = 1,\dots,\delta$\;
 
 Filter to select a basket of $K$ assets\;

 \For{$t = 1,\dots,N$}{
  Choose portfolio $\boldsymbol{\omega}_t = {\left({\omega}_{1,t} , \dots , {\omega}_{K,t}\right)}^\top$\;
  
  Observe $\boldsymbol{R}_t = {\left(R_{1,t} , \dots , R_{K,t}\right)}^\top$\ and receive reward ${\boldsymbol{\omega}_t}^\top \boldsymbol{R}_{t}$\;

 }
 \caption{Sequential Portfolio Selection Problem}
\end{algorithm}

\subsection{Portfolio Construction by Filtering Assets}
\indent

Graph theory has been popularly applied in various disciplines to model networks, where the vertices represent individuals of interest and the edges represent their interactions. For example, in evolutionary game theory, graphs are used to analyze the dynamics of cooperation within different population structures \cite{nowak,fumigration,fustrategy,antisocial,riskygame,conformity}. In financial markets, minimum spanning tree (MST) is accepted as a robust method to visualize the structure of assets \cite{mst}, allowing one to capture different market sectors from empirical data \cite{mantegna,topology,networks}.
\indent

For our purpose, since we have a large pool of assets, we first want to select a basket of $K$ to invest. Recall the return of each asset is $R_{i,t} = \log\left(P_{i,t+1}/P_{i,t}\right)$, where $P_{i,t}$ and $P_{i,t+1}$ are the prices at trial $t$ and $t+1$. Following Mantegna \cite{mantegna} and Mantegna \& Stanley \cite{mantegna2}, we use $\delta$ trials of historical returns to find the correlation matrix, whose entries are $$\rho_{i,j} \stackrel{\text{def}}{=} \frac{\langle H_i H_j \rangle - \langle H_i \rangle \langle H_j \rangle}{\sqrt{( \langle {H_i}^2 \rangle - {\langle H_i \rangle}^2 ) (\langle {H_j}^2 \rangle - {\langle H_j \rangle}^2 )}}$$ where $\langle \cdot \rangle$ is historical mean, namely $\langle H_i \rangle = \sum_{t = 1}^{\delta}H_{i,t}$ for each asset $i$ in the market. For $\delta$ small, we can improve our estimation by taking advantage of the shrinkage method in Ledoit \& Wolf \cite{ledoit}. We then define the metric distance between two vertices as $d_{i,j} \stackrel{\text{def}}{=} \sqrt{2(1 - \rho_{i,j})}$. The Euclidean distance matrix $\boldsymbol{D}$ whose entries are $d_{i,j}$ is then used to compute the undirected graph $G = \{V,E\}$, where $V$ is the set of vertices representing assets and $E$ is the set of weighted edges representing distance. To extract the most important edges from $G$, we construct the minimum spanning tree $T$. In particular, $T$ is the subgraph of $G$ that connects all vertices without cycle and minimizes total edge weights.
\indent

One way to classify vertices is based on their relative positions in the graph, central versus peripheral. In financial markets, this classification method turns out to have significant implications in \textit{systemic risk}, which is the risk that an economic shock causes the collapse of a chain of institutions \cite{schwarcz}. Several empirical studies suggest such risk can be associated with certain characteristics of the correlation structure of market. For example, Kritzman \textit{et al.} \cite{absorption} defines the \textit{absorption ratio} as the fraction of total variances explained by a fixed number of principal components, namely the eigenvectors of the covariance matrix, and shows this ratio increased dramatically during both domestic and global financial crisis including the housing bubble, dot-com bubble, the 1997 Asian financial crisis and so on. Drozdz \textit{et al.} \cite{max} finds a similar result and suggests that the maximum eigenvalue of the correlation matrix rises during crisis and exhausts the total variances. Hence, graph theory can be naturally applied to this setting and provides significant insights in managing systemic risk. In particular, Huang \textit{et al.} \cite{bipartite} gives an intuitive simulation of the contagion process of systemic risk on bipartite graph. Onnela \textit{et al.} \cite{onnela} shows that the minimum spanning tree of assets shrinks during crisis, which supports the above arguments on the compactness of the eigenvalues of correlation matrix. More importantly, Onnela \textit{et al.} \cite{onnela}, Pozzi \textit{et al.} \cite{pozzi} and Ren \textit{et al.} \cite{ren} suggest investing in the assets located on the peripheral parts of the minimum spanning tree can facilitates diversification and reduce the exposure to systemic risk during crisis.

\begin{figure}[!htb]
    \centering
    \includegraphics[scale = 0.72]{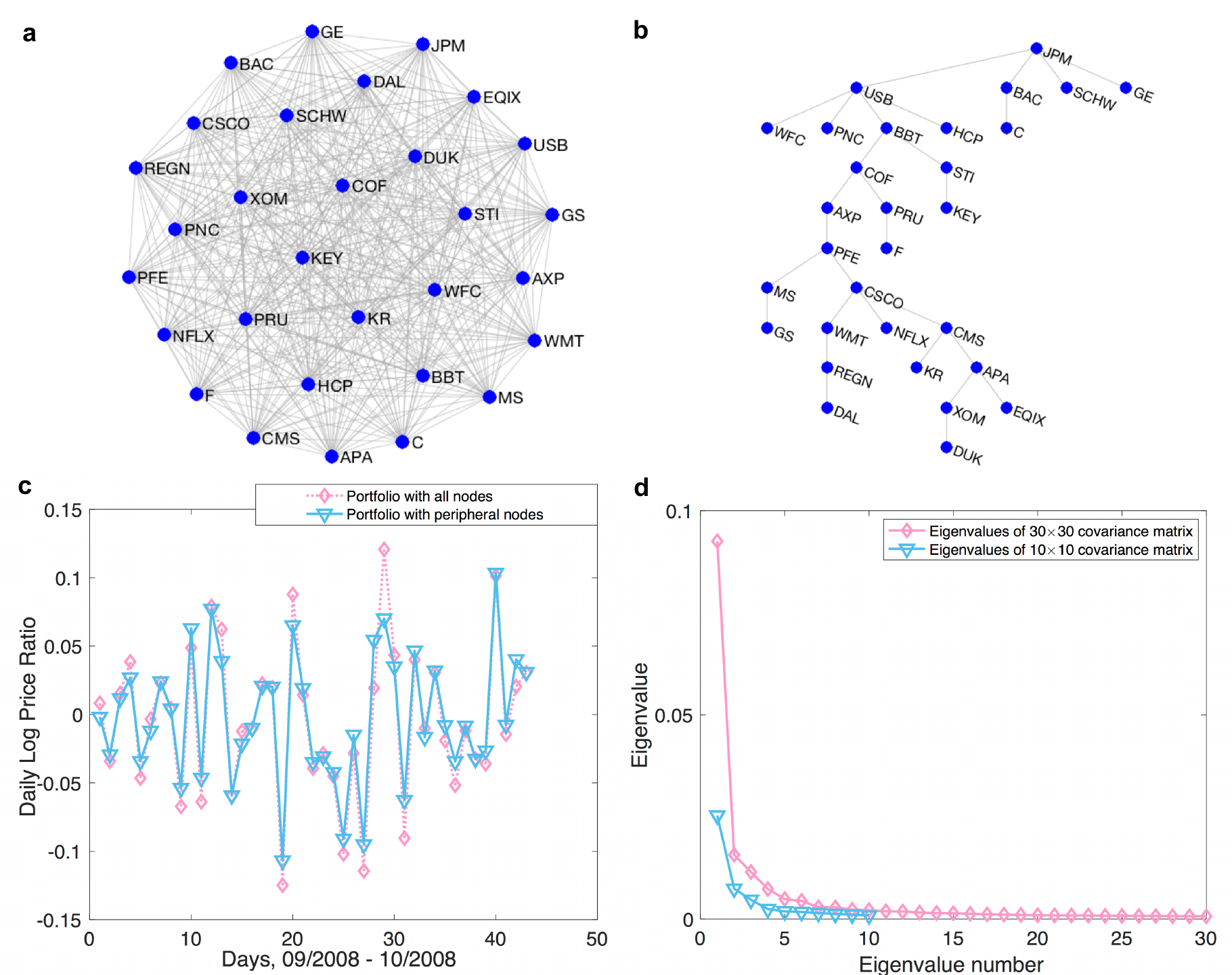}
    \caption{Portfolio selection based on minimum spanning tree. Shown are (a) the complete graph and (b) the corresponding minimum spanning tree constructed from the 30 selected S\&P 500 stocks during the period 09/2008 - 10/2008. Panel (c) plots the performance of the portfolio of 10 randomly selected vertices from the 14 leaves shown in (b). Panel (d) compares the eigenvalue spectrum of the covariance matrix of the 30 selected S\&P 500 stocks in (a) with that of 10 stocks randomly chosen from the peripheral nodes from the minimum spanning tree in (c). }\end{figure}

For our study, we select 30 S\&P 500 stocks, which consist of 15 financial institutions (JPM, WFC, BAC, C, GS, USB, MS, KEY, PNC, COF, AXP, PRU, SCHW, BBT, STI) and 15 randomly selected companies from other sectors (KR, PFE, XOM, WMT, DAL, CSCO, HCP, EQIX, DUK, NFLX, GE, APA, F, REGN, CMS). We use the daily close price of 44 trading days during the subprime mortgage crisis to construct the minimum spanning tree and investigate the advantage of investing in peripheral vertices using the equally-weighted portfolio strategy. Although the number of stocks is small, our results similarly show that investing in peripheral vertices can reduce loss during financial crisis (Fig. 1). Fig.~1(a) shows the complete graph of 30 stocks. Fig.~1(b) is the minimum spanning tree we obtain following the above method. Observe that this tree has a total of 14 leaves (WFC, C, GS, KEY, PNC, SCHW, KR, DAL, HCP, EQIX, DUK, NFLX, GE, F), and selecting from these leaves to construct portfolio almost always reduces the median daily loss compared to the portfolio with all vertices. For example, Fig.~1(c) provides the performance of the portfolio with 10 randomly selected vertices from the 14 leaves, which increases the median daily log price ratio from -0.0101 to -0.0079 and the median daily percentage return from -0.0095 to -0.0070. Furthermore, Fig.~1(d) shows that the eigenvalue spectrum of the covariance matrix becomes less compact. Finally, We acknowledge the dynamic nature of the market structure, but for simplicity this aspect will not be considered in our study.

\indent
Therefore, we select the $K$ most peripheral vertices from the minimum spanning tree $T$ as our basket of assets to invest. We note that for any graph $G$ with distinct edge weight, which is often the case for financial data with high precision, the minimum spanning tree $T$ is proved to be unique. Our selection of vertices tends to lie on the leaves for a star-like graph, on the two ends of the longest edge for a cycle, and on the corners for a lattice. Among the numerous centrality measures discussed in graph theory \cite{freeman}, we use the most straightforward measure and select the $K$ vertices with the least \textit{degree}. The value of $K$ is subjective and can be determined based on the learner's view of the economic state. Assuming $K$ assets are selected, we proceed to portfolio construction as described in what follows.

\subsection{Combined Sequential Portfolio Selection Algorithm}
\indent

We design a sequential portfolio selection algorithm by combining the optimal multi-armed bandit policy, namely the UCB1 proposed in Auer \textit{et al.} \cite{auer}, with the minimization of of a coherent risk measure, namely the conditional value-at-risk. Recall that the return $R_{i,t}$ of each asset $i$ is defined as the log price ratio, namely $R_{i,t} = \log{\left( P_{i,t+1}/P_{i,t} \right)}$. The UCB1 policy is defined as follows. First, select each asset once and observe return during the first $K$ trials. Then, for each trial select the asset that maximizes an estimated upper confidence bound of return with a certain confidence level. Precisely, at each trial $t$ we select
\begin{equation}
            {I_t}^* \hspace{1mm} \stackrel{\text{def}}{=} \hspace{1mm}
            \begin{cases*}
            \hspace{4mm} t & if $t \le K$ \\
            \hspace{1mm} \argmaxA\limits_{\forall i \in \{1,\dots,K\}}  \bar{R}_i (t) + \sqrt{\frac{2\log{t}}{\mathnormal{T}_i \left( t - 1 \right)}} & otherwise
            \end{cases*}
            \label{singlearm}
\end{equation}
where $\bar{R}_i (t)$ is the empirical mean of return for asset $i$ and recall $\mathnormal{T}_i \left( t - 1 \right)$ is the number of times asset $i$ has been selected during the past $t - 1$ trials. Theorem 2.1 below provided in Auer \textit{et al.} \cite{auer} proves the optimality of UCB1.

\begin{theorem}
\textit{(Auer \textit{et al.}, 2002) For all $K > 1$ assets whose the mean returns are in the support $[0,1]$, the regret of UCB1 algorithm after any number $n$ of trials satisfies $$\widehat{\xi} \left( n \right) \le \left[ 8 \sum_{i:\mu_i < \mu^*} \left( \frac{\log{n}}{\mu^* - \mu_i} \right) \right] + \left( 1+ \frac{\pi^2}{3} \right) \left[ \sum_{i = 1}^{K} \left(\mu^* - \mu_i \right) \right]$$ where recall $\mu_i$ is the mean return of asset $i$ and $\mu^* = \max{\{ \mu_1, \dots, \mu_K \}}$}.
\end{theorem}

\indent
The proof makes no assumption on the dependency and distribution of asset returns besides $\mathbb{E}[R_{i,t} | R_{i,1},\dots,R_{i,t-1}] = \mu_i$. Therefore, by scaling the values we can achieve optimality. In addition, we can use historical returns and observed returns of unselected assets to further improve the performance, but we do not discuss details here. Let $\boldsymbol{e}_i \in \mathbb{R}^K$ be the vector of a single $1$ on entry $i$ and $0$ on the others. Our single-asset multi-armed bandit portfolio at $t$ chosen according to Eq.~\eqref{singlearm} is
\begin{equation}
    \boldsymbol{\omega}_t^M \hspace{1mm} \stackrel{\text{def}}{=} \hspace{1mm} \boldsymbol{e}_{{I_t}^*}
        \label{singlearmp}
\end{equation}
\indent
Now, let us incorporate risk-awareness into our algorithm by finding the portfolio that achieves the global minimum of the conditional value-at-risk. We define risk measure and associated properties following Artzner \textit{et al.} \cite{artzner} and Bauerle \& Rieder \cite{mdp}.

\begin{definition}
\textit{Let $(\Omega, \mathcal{F}, \mathbb{P})$ be a probability space and denote by $\mathcal{L} (\Omega, \mathcal{F}, \mathbb{P})$ the set of integrable random variables, where any instance of $\mathcal{L} (\Omega, \mathcal{F}, \mathbb{P})$ represents portfolio return. A function $\Psi : \mathcal{L} (\Omega, \mathcal{F}, \mathbb{P}) \to \mathbb{R}$ is called a risk measure}.
\end{definition}

\begin{definition}
\textit{
Let $\Psi$ be a risk measure, we say $\Psi$ is a coherent risk measure if for all $X_1, X_2 \in \mathcal{L} (\Omega, \mathcal{F}, \mathbb{P})$, $c \in \mathbb{R}$, and $d \in \mathbb{R} \cup \{0\}$, it satisfies
\begin{itemize}
    \item{Translation invariance: $\Psi ( X_1 + c ) = \Psi ( X_1 ) - c$}
    \item{Subadditivity: $\Psi (X_1 + X_2) \le \Psi (X_1) + \Psi (X_2)$}
    \item{Positive homogeneity: $\Psi (dX_1 ) = d\Psi (X_1 )$}
    \item{Monotonicity: $X_1 \le X_2 \Rightarrow \Psi (X_1 ) \ge \Psi (X_2 )$}
\end{itemize}
}
\end{definition}

\begin{definition}
\textit{
Let $X \in \mathcal{L} (\Omega, \mathcal{F}, \mathbb{P})$, the risk measure value-at-risk of $X$ at confidence level $\beta \in (0,1)$ is defined as} $$VaR_{\beta} (X) \stackrel{\text{def}}{=} \inf \{x \in \mathbb{R}: \mathbb{P} (x + X < 0) \le 1 - \beta \}$$ \textit{In addition, the risk measure conditional value-at-risk at confidence level $\gamma \in (0,1)$ is defined as} $$CVaR_{\gamma} (X) \stackrel{\text{def}}{=} \frac{1}{1 - \gamma} \int_{\gamma}^{1} VaR_{\beta} (X) \hspace{1mm} d{\beta}$$
\end{definition}

\indent
In literature, the above risk measures are sometimes expressed in terms of the portfolio loss variable, namely positive values represent loss and negative values represent gain. We note that these definitions are equivalent. Intuitively, The value-at-risk denotes the maximum threshold of loss under a certain confidence level, and conditional value-at-risk is the conditional expectation of loss given that it exceeds such a threshold. Although more popularly used in practice, value-at-risk fails certain mathematical properties such as subadditivity, which contradicts with Markowitz's modern portfolio theory and implies diversification may not reduce investment risk. As a result, it is not a coherent risk measure. On the other hand, Pflug \cite{pflug} proves conditional value-at-risk is coherent and satisfies some extra properties such as convexity, monotonicity with respect to first-order stochastic dominance (FSD) and second-order monotonic dominance.

\begin{theorem}
\textit{
(Pflug, 2000) The conditional value-at-risk is a coherent risk measure.
}
\end{theorem}

\indent
Therefore, we would like to minimize risk using the conditional value-at-risk at confidence level $\gamma$ as the risk measure. We recall that $W = \{\boldsymbol{u} \in \mathbb{R}_{+}^{K} : \boldsymbol{u}^\top \boldsymbol{1} = 1\}$ is the set of possible portfolios. At each trial $t$, the learner would like to solve the following optimization problem
$$
\begin{aligned}
& \underset{\boldsymbol{u} \in W}{\text{minimize}}
& & CVaR_{\gamma} (\boldsymbol{u}^\top \boldsymbol{R}_t )
\end{aligned}
$$
Note that as $\gamma \to 0$, the problem becomes minimizing expected loss and as $\gamma \to 1$ it becomes minimizing the worst outcome. In this study we use $\gamma = 0.95$. Rockafellar \& Uryasev \cite{uryasev} provides a convenient method to solve this problem. Recall that we assume both historical returns and present returns follow the same distribution, let $p(\boldsymbol{R}_t)$ be the density. Define the performance function as $$F_\gamma (\boldsymbol{u},\alpha) \hspace{1mm} \stackrel{\text{def}}{=} \hspace{1mm} \alpha + \frac{1}{1 - \gamma}\int_{\boldsymbol{\boldsymbol{R}_t} \in \mathbb{R}^K} {\left[ - \boldsymbol{u}^\top \boldsymbol{R}_t - \alpha \right] }^{+} p(\boldsymbol{R}_t) \hspace{1mm} d\boldsymbol{R}_t$$ where $[m]^+ \stackrel{\text{def}}{=} \max\{m,0\}$. We have the following theorem proved in Rockafellar \& Uryasev \cite{uryasev}.

\begin{theorem}
\textit{
(Rockafellar \& Uryasev, 2000) The minimization of $\hspace{1mm} CVaR_\gamma (\boldsymbol{u}^\top \boldsymbol{R}_t )$ over $\boldsymbol{u} \in W$ is equivalent to the minimization of $F_\gamma (\boldsymbol{u},\alpha )$ over all pairs of $(\boldsymbol{u},\alpha ) \in W \times \mathbb{R}$. Moreover, since $F_\gamma (\boldsymbol{u},\alpha )$ is convex with respect to $(\boldsymbol{u},\alpha )$, the loss function $-\boldsymbol{u}^\top \boldsymbol{R}_t$ is convex with respect to $\boldsymbol{u}$ and $W$ is a convex set due to linearity, the minimization of $F_\gamma (\boldsymbol{u},\alpha )$ is an instance of convex programming.
}
\end{theorem}

\indent
Moreover, since the density $p(\boldsymbol{R}_t)$ is unknown, we would like to approximate the performance function using not only historical returns but also knowledge gained as we proceed in this learning process. From the received $H_{i,1},\dots,H_{i,\delta}$ for all $i$, we extract historical returns of our $K$ assets $\boldsymbol{H}_1,\dots,\boldsymbol{H}_\delta \in \mathbb{R}^K$. Let $\boldsymbol{R}_1,\dots,\boldsymbol{R}_{t-1}$ be the $t-1$ trials of returns observed so far, then our approximation of $F_\gamma (\boldsymbol{u},\alpha)$ at trial $t$ is the following convex and piecewise linear function

\begin{equation}
\widetilde{F}_\gamma (\boldsymbol{u},\alpha,t) \stackrel{\text{def}}{=} \alpha + \frac{1}{(\delta + t - 1) (1 - \gamma)} \Bigg[ \sum_{s = 1}^{\delta} {\big[ - \boldsymbol{u}^\top \boldsymbol{H}_s - \alpha \big] }^{+} + \sum_{s = 1}^{t - 1} {\big[ - \boldsymbol{u}^\top \boldsymbol{R}_s - \alpha \big] }^{+} \Bigg].
\label{riskaware}
\end{equation}
Notice that the approximation function is implicitly also a function of the current trial $t$, hence we have added an extra parameter and denote it as $\widetilde{F}_\gamma (\boldsymbol{u},\alpha, t)$. As the learner proceeds in time, she accumulates data information and obtains an increasingly more precise approximation. As a result, the minimization of conditional value-at-risk is solved by convex programming and generates the following optimal solution. At each trial $t$, the risk-aware portfolio constructed according to Eq.~\eqref{riskaware} is
\begin{equation}
    \boldsymbol{\omega}_t^C \hspace{1mm} \stackrel{\text{def}}{=} \hspace{1mm} \argminA\limits_{(\boldsymbol{u},\alpha ) \in W \times \mathbb{R}} \widetilde{F}_\gamma (\boldsymbol{u},\alpha,t)
    \label{riskawarep}
\end{equation}

\indent
Now we have found both the single-asset multi-armed bandit portfolio by~\eqref{singlearmp} and the risk-aware portfolio by~\eqref{riskawarep}. Notice that they are dynamic and update based on the learner's accumulated knowledge. For each trial $t$, the learner combines them with a factor $\lambda \in [0,1]$ to form the balanced portfolio
\begin{equation}
    \boldsymbol{\omega}_t^* \hspace{1mm} \stackrel{\text{def}}{=} \hspace{1mm} \lambda \boldsymbol{\omega}_t^M + (1-\lambda)\boldsymbol{\omega}_t^C
    \label{combinedp}
\end{equation}
In particular, $\lambda$ is the proportion of wealth invested in the single-asset multi-armed bandit portfolio and  $1-\lambda$ is the proportion invested in the risk-aware portfolio. The value of $\lambda$ denotes the risk preference of the learner. As $\lambda \to 1$, our algorithm reverts to the UCB1 policy, whereas for $\lambda \to 0$, it becomes the minimization of conditional value-at-risk. Therefore, the commonly discussed trade-off between reward and risk is illustrated here in the choice of $\lambda$. Finally, Algorithm 1 below summarizes our sequential portfolio selection algorithm.

\setcounter{algocf}{0}
\begin{algorithm}[h]
\SetAlgoLined
\KwIn{$K$, $\gamma$, $\lambda$}

 Select $K$ peripheral assets from the market according to Section 3.1\;

 \For{$t = 1,\dots,N$}{
  Compute the single-asset multi-armed bandit portfolio $\boldsymbol{\omega}_t^M$ by~\eqref{singlearmp}\;
  
  Compute the risk-aware portfolio $\boldsymbol{\omega}_t^C$ at confidence level $\gamma$ by~\eqref{riskawarep}\;
  
  Select the combined portfolio $\boldsymbol{\omega}_t^*$ with a factor $\lambda$ by~\eqref{combinedp}\;
  
  Observe returns $\boldsymbol{R}_t$ and update accumulated knowledge for~\eqref{singlearmp}and~\eqref{riskawarep}\;
  
  Receive portfolio reward ${\boldsymbol{\omega}_t^*}^\top \boldsymbol{R}_{t}$\;

 }
 \caption{Our Proposed Sequential Portfolio Selection Algorithm}
\end{algorithm}

\section{Results}
\indent

In this Section, we design experiments and report the performance of the proposed algorithm (see Algorithm 1) in comparison with several benchmarks. 

\subsection{Monte Carlo Simulation Method}
\indent

For simplicity, we consider stocks as our assets and adopt the Black-Scholes model \cite{blackscholes} to simulate stock prices as geometric Brownian motion (GBM) paths. As a Nobel Prize winning model, it provides a partial differential equation to price an European option by computing the initial wealth for perfectly hedging a short position in that option. The underlying asset, usually a stock, is modeled to follow a geometric Brownian motion. Although this assumption may not hold perfectly in reality, it provides an extremely convenient and popularly used method to simulate any number of stock paths. For our purpose, since we never make any assumption on the dependency of asset returns, we consider the general case where stock paths can be correlated as it is almost always the case in financial market. We use definitions similar to Chapter 4 of Shreve \cite{shreve} and describe our method below.

\begin{definition}
\textit{
Let $(\Omega, \mathcal{F}, \mathbb{P})$ be a probability space. The stock price $P_i (t)$ is said to follow a geometric Brownian motion if it satisfies the following stochastic differential equation  $$d P_i (t) = {\alpha_i} P_i (t) \hspace{0.5mm} dt + {\sigma_i} P (t) \hspace{0.5mm} dW_i(t)$$ where $W_i(t)$ is a Brownian motion, $\alpha_i$ is drift and $\sigma_i$ is volatility.
}
\end{definition}

\begin{definition}
\textit{
Two stock paths $P_i (t)$ and $P_j (t)$ modeled by geometric Brownian motions are correlated if their associated Brownian motions satisfy $$dW_i(t) \hspace{0.5mm} dW_j(t) = \rho_{i,j} \cdot dt$$ for some nonzero constant $\rho_{i,j} \in [-1,1]$ where $\rho_{i,i} = \rho_{j,j} = 1$.
}
\end{definition}

\begin{proposition}
\textit{
For two correlated stock prices $P_i(t)$ and $P_j(t)$ that satisfy $dW_i(t) dW_j(t) = \rho_{i,j} \cdot dt$, the following properties hold:
\begin{itemize}
    \item{$\mathbb{E}[W_i(t) W_j(t)] = \rho_{i,j} \cdot t$}
    \item{$\Cov[W_i(t),W_j(t)] = \rho_{i,j} \cdot t$}
    \item{$\Cov[\sigma_i W_i(t),\sigma_j W_j(t)] = \sigma_i \sigma_j \rho_{i,j} \cdot t$}
\end{itemize}
where $\sigma_i$ and $\sigma_j$ are volatility parameters of $P_i(t)$ and $P_j(t)$ respectively.
}
\end{proposition}
\begin{proof}
We prove the first claim and the rest follow immediately after some computations. By It\^{o}-Doeblin formula, which can be found in Shreve \cite{shreve}, we have
$$
d(W_i(t) W_j(t)) = W_i(t) \hspace{0.5mm} dW_j(t) + W_j(t) \hspace{0.5mm} dW_i(t) + \rho_{i,j} \cdot dt
$$
Integrate on both sides, we have 
$$
W_i(t) W_j(t) = \int_{0}^{t}W_i(t) \hspace{0.5mm} dW_j(t) + \int_{0}^{t}W_j(t) \hspace{0.5mm} dW_i(t) + \rho_{i,j} \cdot t
$$
By the martingale property of it\^{o} integrals, we simply take the expectation on both sides to obtain $\mathbb{E}[W_i(t) W_j(t)] = \rho_{i,j} \cdot t$.
\end{proof}

Recall that we have $K$ stocks whose prices $P_1(t),\dots,P_K(t)$ are modeled by correlated geometric Brownian motions. By definition, they must satisfy the following two equations
\begin{equation}
\frac{dP_i(t)}{P_i(t)} = \alpha_i \hspace{0.5mm} dt + \sigma_i \hspace{0.5mm} dW_i(t),
\label{Gmotion}
\end{equation}
and
\begin{equation}
dW_i(t) \hspace{0.5mm} dW_j(t) = \rho_{i,j} \cdot dt
\label{Gmotion2}
\end{equation}
In particular, the solution to Eq.~\eqref{Gmotion} can be expressed as follows \cite{glasserman}. For any time $u < l$ we have
\begin{equation}
P_i(l) = P_i(u) \cdot \exp\{ (\alpha_i - \frac{1}{2}{\sigma_i}^2)(l-u) + \sigma_i (W_i(l) - W_i(u)) \}
\label{GmotionSolution}
\end{equation}
We first would like to express the scaled correlated Brownian motions $\sigma_i W_i(t)$ using independent ones. By Proposition 3.1, we have the following instantaneous covariance matrix $$ \Theta = 
\begin{bmatrix}
    {\sigma_1}^2 & \sigma_1 \sigma_2 \rho_{1,2}  & \dots  & \sigma_1 \sigma_K \rho_{1,K}  \\
    \sigma_2 \sigma_1 \rho_{2,1}  & {\sigma_2}^2  & \dots  & \sigma_2 \sigma_K \rho_{2,K}  \\
    \vdots & \vdots  & \ddots & \vdots \\
    \sigma_K \sigma_1 \rho_{K,1}  & \sigma_K \sigma_2 \rho_{K,2}   & \dots  & {\sigma_K}^2
\end{bmatrix}
$$
Since $\Theta$ has to be symmetric and positive definite, it has a square root and we apply Cholesky decomposition to find the matrix $A$ such that $AA^T = \Theta$. By Shreve \cite{shreve}, there exists $K$ independent Brownian motions $X_1 (t), \dots, X_K(t)$ such that $$\sigma_i W_i(t) = \sum_{m = 1}^{K} A_{i,m} X_m (t)$$
Then Eq.~\eqref{Gmotion} becomes
\begin{equation}
\frac{dP_i(t)}{P_i(t)} = \alpha_i \hspace{0.5mm} dt + \sum_{m = 1}^{K} A_{i,m} \hspace{0.5mm} dX_m (t)
\end{equation}
and Eq.~\eqref{GmotionSolution} becomes that for any time $u < l$
\begin{equation}
P_i (l) = P_i (u) \exp\{ (\alpha_i - \frac{1}{2}{\sigma_i}^2)(l-u) + \sum_{m =1}^{K} A_{i,m} (X_m(l) - X_m(u)) \}
\label{GmotionSolution2}
\end{equation}
Since each Brownian motion $X_m (t)$ for $m \in [1,K]$ above is independent and the increment $X_m (l) - X_m(u)$ is Gaussian with mean $0$ and variance $l - u$, let $\boldsymbol{Z}(t) = (Z_1(t),\dots,Z_K(t))^\top$ be standard multivariate Gaussian then Eq.~\eqref{GmotionSolution2} becomes
\begin{equation}
P_i (l) = P_i (u) \exp\{ (\alpha_i - \frac{1}{2}{\sigma_i}^2)(l-u) + \sqrt{l-u} \sum_{m =1}^{K} A_{i,m} Z_m(l) \}
\label{GmotionSolution3}
\end{equation}
Therefore, at each time we can conveniently generate a sample from $\boldsymbol{Z}(t)$ to compute the price increment. Specifically, Eq.~\eqref{GmotionSolution3} leads to the following recursive algorithm that can also be found in Glasserman \cite{glasserman}. For $0 = t_0 < t_1 < \cdots < t_{\infty}$ we have $$P_i(t_{s+1}) = P_i(t_s) \cdot \exp\{ (\alpha_i - \frac{1}{2}{\sigma_i}^2)(t_{s+1}-t_s) + \sqrt{t_{s+1} - t_s} \sum_{m =1}^{K} A_{i,m} Z_m(t_{s+1}) \}$$ Also notice that when the paths are independent, $dW_i(t) \hspace{0.5mm} dW_j(t) = \boldsymbol{\delta}_{i,j} \hspace{0.5mm} dt$ where $\boldsymbol{\delta}_{i,j}$ is the Kronecker delta function, and the covariance matrix $\Theta$ is diagonal. In this special case, it is equivalent to compute $K$ paths separately in the one-dimensional space. For our purpose, we first find some appropriate covariance matrix and generate $K$ price paths following the above algorithm. We then uniformly divide the total time horizon into $\delta + N$ trials and use the price at the beginning and end of each trial to calculate return, which is defined earlier as the log price ratio. We run our sequential portfolio selection algorithm on these data and compare the performance with four benchmark portfolios, namely UCB1~\eqref{singlearmp}, risk-aware portfolio~\eqref{riskawarep}, $\epsilon$-greedy and the equally-weighted portfolio.

\subsection{Simulation results}
\indent
\begin{figure}[!htb]
    \centering
    \includegraphics[scale = 0.76]{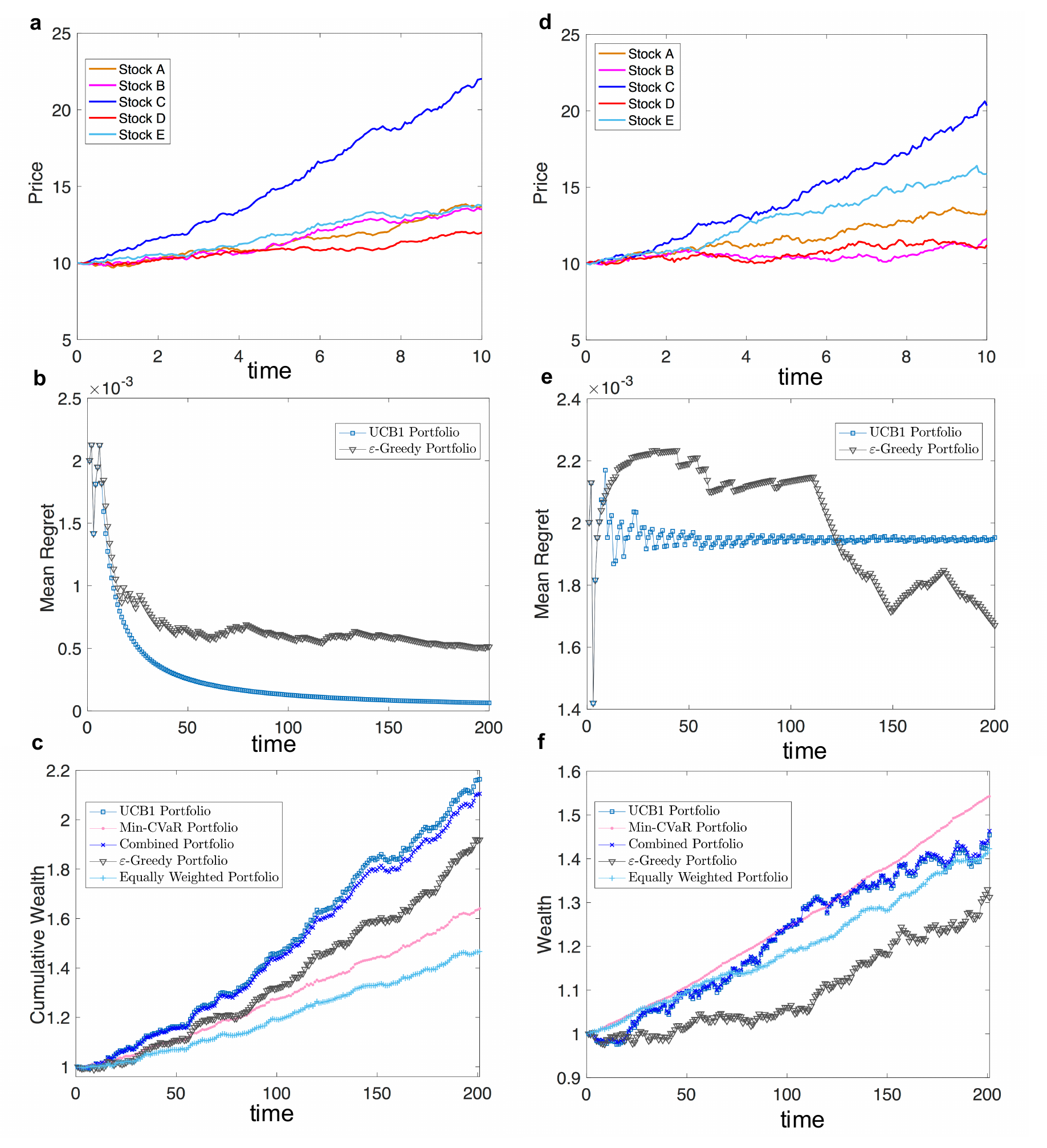}
    \caption{Combined sequential portfolio selection algorithm can achieve a balance between risk and return. Panel (a) and (c) show the simulated stock paths based on the geometric Brownian motion. Panel (b) and (d) plot the performance of two portfolio selection algorithms, UCB1 vs. $\varepsilon$-Greedy. Panel (c) and (e) compare the cumulative wealth obtained with our sequential portfolio selection algorithm that combines the single-asset multi-armed bandit portfolio by~\eqref{singlearmp} and the risk-aware portfolio by~\eqref{riskawarep} with the other four benchmarks of portfolio selection algorithms. To quantify and compare the role of volatility in the performance of portfolio selection algorithms, we present the simulation results of low volatility in left panels (a)(b)(c) and high volatility in right panels (d)(e)(f). Parameters: The same vector $(0.04, 0.035, 0.08, 0.02, 0.03)$ for drift terms $\alpha_i$ is used for simulating the stock paths in (a) and (d). For each trial, the volatility terms $\sigma_i$ are uniformly and randomly generated from the interval $[0.02,0.025]$ in (a) and from the interval $[0.03,0.035]$ in (d). $\lambda = 0.9$. }
\end{figure}

After we repeatedly generate price paths and compare the performance, we can see the results agree well with our prediction (Fig. 2). The UCB1 portfolio almost always achieves the most cumulative wealth but has high variations in its path. On the other hand, the risk-aware portfolio achieves a relatively low cumulative wealth but also has low variations. As a result, our combined portfolio achieves a middle ground between the two extremes of maximizing reward and minimizing risk. For example, Fig.~2a-2c illustrate a typical simulation, where Fig.~2a shows $K = 5$ geometric Brownian motion paths, Fig.~2b shows the optimality of UCB1 compared to $\epsilon$-greedy, Fig.~2c shows the cumulative wealth at the end of $N = 200$ trials. With an initial wealth of 1 and $\lambda = 0.9$, the cumulative wealth is 2.1615 for UCB1, 2.1024 for combined portfolio, 1.9168 for $\epsilon$-greedy, 1.6355 for risk-aware portfolio, and 1.4640 for the equally-weighted portfolio.

In addition, we observe that when the market is volatile and when different stock paths are similar in expectation, it takes more trials for the UCB1 policy to reach optimality(Fig.~2d-2f). In this case, the risk-aware portfolio achieves the most cumulative wealth with a similarly low variation in its path. Different from the simulation presented in Fig.~2a-2c, where the volatility parameters of geometric Brownian motions are bounded in the interval $[0.02, 0.025]$, we now choose values from the interval $[0.03, 0.035]$ for Fig.~2d-2f. Specifically, Fig.~2d-2f demonstrate such a simulation, where Fig.~2d shows the geometric Brownian motion paths, Fig.~2e shows the suboptimality of UCB1, and Fig.~2f shows the cumulative wealth at the end of 200 trials. With an initial wealth of 1 and $\lambda = 0.9$, the cumulative wealth is 1.5412 for risk-aware portfolio, 1.4409 for combined portfolio, 1.4294 for UCB1, 1.4132 for the equally-weighted portfolio, and finally 1.3298 for $\epsilon$-greedy.

From the above discussion, it is evident that the value of $\lambda$ is vital to the performance of our sequential portfolio selection algorithm and should be determined based on market condition. In particular, Way \textit{et al.} \cite{way} discusses the trade-off between specialization to achieve high rewards and diversification to hedge against risk, and similarly shows that such choice depends on the underlying parameters and initial conditions.

\section{Discussion and Conclusions}

In this paper, we have studied the multi-armed bandit problem as a mathematical model for sequential decision making under uncertainty. In particular, we focus on its application in financial markets and construct a sequential portfolio selection algorithm. We first apply graph theory and select the peripheral assets from the market to invest. Then at each trial, we combine the optimal multi-armed bandit policy with the minimization of a coherent risk measure. By adjusting the parameter, we are able to achieve the balance between maximizing reward and minimizing risk. We adopt the Black-Scholes model to repeatedly simulate stock paths and observe the performance of our algorithm. We conclude that the results agree well with our prediction when the market is stable. In addition, when the market is volatile, risk-awareness becomes more crucial to achieving high performance. Therefore, parameter selection should be based on the market condition.
 
\indent
For future research, one may consider the optimal selection of the parameter $\lambda$ for combining the two portfolios. One may also consider portfolio selection strategies based on the Markov Decision Process, which is a generalization of the multi-armed bandit to multiple states. In addition, one may pay more attention to a chaotic market environment where stock paths can be affected by various factors instead of simply following a stochastic process. For example, Junior \& Mart \cite{news} uses random matrix theory and transfer entropy to show news articles can possibly affect the market. Finally, one may consider transaction costs and market liquidity. For example, Reiter \textit{et al.} \cite{auction} illustrates the trade-off between reward and cost in a biological auction setting and might provide some important insights for the researcher.

\section*{Data Availability}
The datasets generated during and/or analysed during the current study are available from the corresponding author on reasonable request.

\section*{Acknowledgements}
X.H. is thankful for financial support from the National Science Foundation and Dartmouth College. F.F. is grateful for support from the Dartmouth Faculty Startup Fund, Walter \& Constance Burke Research Initiation Award, NIH under grant no. C16A12652 (A10712), and DARPA under grant no. D17PC00002-002.

\section*{Author Contributions}
X.H. \& F.F. conceived the project, X.H. performed analyses and simulations, X.H. \& F.F. analyzed results, and X.H. wrote the first draft of main text. All authors reviewed the manuscript.

\section*{Additional Information}
The authors declare no competing financial interests.


\end{document}